\documentclass[11pt]{amsart}

\usepackage{accents}
\usepackage{appendix}
\usepackage{amsfonts}
\usepackage{amsmath}
\usepackage{amssymb}	
\usepackage{amsthm}
\usepackage{array,booktabs,multirow}
\usepackage{braket}
\usepackage{cite}
\usepackage{clrscode3e}
\usepackage{dsfont}
\usepackage[shortlabels]{enumitem}
\usepackage{etoolbox}
\usepackage{float}
\usepackage[hang, flushmargin]{footmisc}
\usepackage{latexsym}
\usepackage{lipsum}
\usepackage{tikz}
\usepackage{hyperref}
\usetikzlibrary{matrix,arrows}

\usepackage{comment}

\theoremstyle{plain}
\newtheorem{thm}{Theorem}[section]

\newtheorem{lem}[thm]{Lemma}

\theoremstyle{definition}

\theoremstyle{remark}

\setlist[enumerate,1]{leftmargin=2em}

\def\e{\varepsilon}

\begin{document}

\title[An algorithm to evaluate the spectral expansion]{An algorithm to evaluate the spectral expansion}

\author{Hau-Wen Huang}
\address{
Department of Mathematics\\
National Central University\\
Chung-Li 32001 Taiwan
}
\email{hauwenh@math.ncu.edu.tw}

\begin{abstract}
Assume that $X$ is a connected $(q+1)$-regular undirected graph of finite order $n$. Let $A$ denote the adjacency matrix of $X$. Let $\lambda_1=q+1>\lambda_2\geq \lambda_3\geq \ldots \geq \lambda_n$ denote the eigenvalues of $A$.  The spectral expansion of $X$ is defined by 
$$
\Delta(X)=\lambda_1-\max_{2\leq i\leq n}|\lambda_i|.
$$
By the Alon--Boppana theorem, when $n$ is sufficiently large, $\Delta(X)$ is quite high if 
$$
\mu(X)=q^{-\frac{1}{2}}
\max_{2\leq i\leq n}|\lambda_i|
$$ is close to $2$. In this paper, with the inputs $A$ and a real number $\e>0$ we design an algorithm to estimate if $\mu(X)\leq 2+\e$ in $O(n^\omega \log \log_{1+\e} n )$ time, where $\omega<2.3729$ is the exponent of matrix multiplication.
\end{abstract}
\maketitle



\allowdisplaybreaks

\section{Introduction}


Throughout this paper we adopt the following conventions: 
The function $\log$ without a base refers to base $2$.
The notation $\omega$ stands for the exponent of matrix multiplication. 
The improvements of the Coppersmith--Winograd algorithm \cite{CWalgorithm:1990,CWalgorithm:2012,CWalgorithm:2014} indicate that 
$\omega<2.3729$. 
Let $q \geq 1$ and $n \geq 2$ be two integers.
We assume that $X$ is a connected $(q+1)$-regular undirected graph with $n$ vertices. 
Let $A$ denote the adjacency matrix of $X$. 
Let 
$$
\lambda_1=q+1>\lambda_2\geq \lambda_3\geq \ldots \geq \lambda_n
$$
denote the eigenvalues of $A$. 

Roughly speaking, the graph $X$ is a {\it good expander} if the parameter $q$ is low and the expansion parameters are high. Such graphs are widely studied in mathematics and have several applications to computer science such as 
networks, error-correcting codes and probabilistic algorithms \cite{expander2006,Terras2010,Ramanujan2003,Expander1994, Expander1986}. 
There are two common expansion parameters called the spectral expansion and the edge expansion. The {\it spectral expansion} of $X$ is measured by the spectral gap of $X$:
$$
\Delta(X)=\lambda_1-\max_{2\leq i\leq n} |\lambda_i|.
$$
Given any two sets $S,T$ of vertices of $X$, let 
$E(S,T)$ denote the set of all edges with one vertex in $S$ and the other vertex in $T$. The {\it edge expansion} of $X$ is given by 
$$
h(X)=\min_{1\leq |S|\leq n/2}\frac{|E(S,\overline{S})|}{|S|}
$$
where $\min$ is over all sets $S$ of vertices of $X$ with $1\leq |S|\leq \frac{n}{2}$. The expander mixing lemma \cite{Expander-Mix1988} implies that $|E(S,T)|$ will be closer to the expected number $\frac{(q+1)|S| |T|}{n}$ for all sets $S,T$ of vertices of $X$, provided that $\Delta(X)$ is as large as possible. In \cite{edge-ratio1984} an inequality states that  
$$
\frac{q-\lambda_2+1}{2}\leq h(X)\leq \sqrt{2(q+1)(q-\lambda_2+1)}.
$$
We set 
$$
\mu(X)=
q^{-\frac{1}{2}}\max_{2\leq i\leq n} |\lambda_i|.
$$
The Alon--Boppana theorem \cite{Alon-Boppana1991,Alon-Boppana1993} implies that 
$$
\mu(X)\geq q^{-\frac{1}{2}} \lambda_2 \geq 2-o(1)
$$
where the term $o(1)$ is a quantity that tends to $0$ for every fixed $q$ as $n$ approaches $\infty$. Thus, when $n$ is sufficiently large, the expansion parameters $\Delta(X)$ and $h(X)$ are quite high provided that $\mu(X)$ is close to $2$.

To obtain $\mu(X)$, one may compute $\lambda_i$ for all $i=2,3,\ldots,n$ via the QR algorithm or Jacobi eigenvalue algorithm. Each iteration step of both algorithms has $O(n^3)$ time complexity \cite{Demmel1997}. In this paper, we introduce a certain number sequence $\{H_k\}_{k=1}^\infty$ and investigate its connection with $\mu(X)$. Furthermore, we design an algorithm to estimate if $\mu(X)$ is close to $2$ via the sequence $\{H_k\}_{k=1}^\infty$ in $o(n^3)$ time. 
To define $\{H_k\}_{k=1}^\infty$ we begin by recalling the notation of geodesic cycles. We endow two opposite orientations on all edges of $X$ called the {\it oriented edges} of $X$.
Recall that a {\it walk}
$$
P=(e_1,e_2,\ldots,e_k)
$$ 
is a sequence of oriented edges of $X$ such that the end vertex of $e_i$ is the start vertex of $e_{i+1}$ for all $i=1,2,\ldots,k-1$. 
A walk $P$ is called a {\it cycle} if the start vertex of $e_1$ is the end vertex of $e_k$. 
A walk $P$ is said to have {\it backtracking} if $e_{i+1}$ is the oriented edge opposite to $e_i$ for some $i=1,2,\ldots,k-1$. A walk is {\it backtrackless} if it has no backtracking. 
A cycle is {\it geodesic} if its all shifted cycles are backtrackless \cite{Winnie2018,Winnie2019}. Let $N_k$ denote the number of geodesic cycles on $X$ of length $k\geq 1$. The numbers $\{N_k\}_{k=1}^\infty$ first appeared in the {\it Ihara zeta function} $Z(u)$ of $X$ which is defined as the analytic continuation of
\begin{gather*}
\exp
\left(\sum_{k=1}^\infty \frac{N_k}{k} u^k
\right).
\end{gather*}
The Ihara zeta function $Z(u)$ of $X$ was first considered by Y. Ihara \cite{Ihara:1966} in the context of discrete groups. As suggested by J.-P. Serre, $Z(u)$ has a graph-theoretical interpretation \cite{Serre1980}. 
The graph $X$ is called {\it Ramanujan} \cite{Ramanujan1988} whenever 
$$
\max_{|\lambda_i|<q+1}|\lambda_i|\leq 2 q^{\frac{1}{2}}
$$
where $\max$ is over all $i=1,2,\ldots,n$ with $|\lambda_i|<q+1$.
It was discovered by T. Sunada \cite{Sunada:1986} that $Z(u)$ satisfies an analogue of the Riemann hypothesis if and only if $X$ is a Ramanujan graph. 
A recent result \cite{Huang:Li_criterion} provides a necessary and sufficient condition for $X$ as Ramanujan in terms of Hasse--Weil bounds on $\{N_k\}_{k=1}^\infty$. Inspired by \cite{Huang:Li_criterion} we consider the numbers 
\begin{gather}\label{Hk}
H_k
=
\left\{
\begin{array}{ll}
2(n-1)+q^\frac{k}{2}+q^{-\frac{k}{2}}-q^{-\frac{k}{2}} N_k
\qquad 
&\hbox{if $k$ is odd},
\\
2(n-1)+q^\frac{k}{2}+q^{-\frac{k}{2}}-q^{-\frac{k}{2}} (N_k-n(q-1))
\qquad 
&\hbox{if $k$ is even}
\end{array}
\right.
\end{gather}
for all integers $k\geq 1$.

The paper will proceed as follows. In \S\ref{s:compute_NH} we design a dynamic programming algorithm that computes $N_k$ in $O(n^\omega \log k)$ time. Moreover we modify the algorithm to compute $H_k$ with the same time complexity. In \S\ref{s:proof} we relate $\{H_k\}_{k=1}^\infty$ to $\mu(X)$. In \S\ref{s:algorithm},  with the inputs $A$ and a real number $\e>0$ we describe our main algorithm which evaluates if 
$\mu(X)\leq 2+\e$ in $O(n^\omega \log \log_{1+\e} n )$ time.

\section{The computations of $N_k$ and $H_k$}\label{s:compute_NH}


Recall that the {\it directed edge matrix} $W$ of $X$ is the $(0,1)$-matrix indexed by the oriented edges of $X$ such that 
$
W_{ef}=1$ 
if and only if the end vertex of $e$ is the start vertex of $f$ and $e$ is not the oriented edge opposite to $f$
for all oriented edges $e,f$ of $X$. 
Observe that 
\begin{gather}\label{trace}
N_k={\rm trace}(W^k)
\qquad 
\hbox{for all integers $k\geq 1$}.
\end{gather}
Note that (\ref{trace}) holds for irregular graphs. 
Let 
$
m=n(q+1)
$
denote the number of oriented edges of $X$.
In other words, this establishes the following formula \cite{Basse:1992,Terras2010,Hashimoto:1989}:
\begin{gather}\label{Bass}
Z(u)^{-1}=\det(I_{m}-Wu).
\end{gather}
To obtain $W^k$, it only uses at most $2 \lfloor \log k\rfloor$ times of $m\times m$ matrix multiplication by applying binary exponentiation. Thus, along the vein (\ref{trace}) the computation of $N_k$ takes $O(m^\omega \log k)$ time.

For the rest of this paper, let  
$\alpha_i$ denote a root of 
$$
u^2-q^{-\frac{1}{2}}\lambda_i u+1 
\qquad \hbox{for all $i=1,2,\ldots,n$}.
$$  
In addition to (\ref{Bass}) the Ihara zeta function $Z(u)$ of $X$ has the following celebrated formula \cite{Basse:1992,Ihara:1966}:
\begin{gather}\label{Ihara}
Z(u)^{-1}=(1-u^2)^{\frac{n(q-1)}{2}}\det(I_n-Au+q I_n u^2 ).
\end{gather}
We change the variable $u$ to $q^{-\frac{1}{2}} u$ and then take logarithm and differential on either side of (\ref{Ihara}). Multiplying the resulting equation by $u$ yields that $\sum_{k=1}^\infty q^{-\frac{k}{2}} N_k u^{k}$ is equal to $\frac{n(q-1)}{2}\sum_{k=1}^\infty (q^{-\frac{k}{2}}+(-q)^{-\frac{k}{2}}) u^{k}$ plus 
\begin{gather}\label{Nk_2ndpart}
\sum_{i=1}^n \sum_{k=1}^\infty (\alpha_i^k+\alpha_i^{-k}) u^k.
\end{gather}

Define a family of polynomials $\{T_k(x)\}_{k=0}^\infty$ by 
\begin{gather}\label{T}
T_{k+1}(x)=x T_k(x)-T_{k-1}(x) 
\qquad 
\hbox{for all $k\geq 1$}
\end{gather}
with $T_0(x)=2$ and $T_1(x)=x$. Note that $\frac{1}{2} T_k(2 x)$ is the $k$th Chebyshev polynomial of the first kind for all integers $k\geq 0$. Using (\ref{T}) a routine induction shows that 
\begin{gather}\label{T(x+xinv)}
T_k(x+x^{-1})=x^k+x^{-k}
\qquad 
\hbox{for all integers $k\geq 0$}.
\end{gather}
It follows from (\ref{T(x+xinv)}) that (\ref{Nk_2ndpart}) is equal to 
\begin{gather*}
\sum_{i=1}^n  \sum_{k=1}^\infty T_k(q^{-\frac{1}{2}} \lambda_i) u^k
=
\sum_{k=1}^\infty \sum_{i=1}^n T_k(q^{-\frac{1}{2}} \lambda_i) u^k
=
\sum_{k=1}^\infty {\rm trace}(T_k(q^{-\frac{1}{2}} A)) u^k.
\end{gather*}
Hence we obtain that 
\begin{gather}\label{Nk&Tk}
N_k=\left\{
\begin{array}{ll}
q^{\frac{k}{2}} {\rm trace}(T_k(q^{-\frac{1}{2}} A))
\qquad 
&\hbox{if $k$ is odd},
\\
n(q-1)+q^{\frac{k}{2}} {\rm trace}(T_k(q^{-\frac{1}{2}} A))
\qquad 
&\hbox{if $k$ is even}.
\end{array}
\right.
\end{gather}
Based on the formula (\ref{Nk&Tk}), we develop a more efficient algorithm to compute $N_k$. The pseudocode is as follows:

\bigskip

\begin{codebox}
\Procname{$\proc{NGC}(A,k)$}

\li $n\gets \hbox{the number of rows of $A$}$

\li $q\gets \hbox{(a row sum of $A$)}-1$

\li $L\gets \proc{Indices}(k)$

\li $l\gets \hbox{the length of $L$}$
\label{NGC-ell}

\li let $T[0\twodots 3]$ denote a new array with $T[0]\gets A$
\label{NGC-T-initial}

\li let $Q[0\twodots 3]$ denote a new array with $Q[0]\gets 1$ \label{NGC-Q-initial}

\li \For $i\gets l-1$ \Downto $1$    \label{NGC-for-begin}
\li \Do
        \For $j\gets 3$ \Downto $1$   \label{NCG-move-begin}
\li     \Do         
           $Q[j]\gets Q[j-1]$
\li        $T[j]\gets T[j-1]$           
        \End                          \label{NCG-move-end}
\li        \If $L[i-1]$ is even       \label{NCG-even-begin}
\li     \Then  
            \If $L[i-1]\isequal 2 L[i]$   \label{NCG-even-j-begin}
\li         \Then 
                $j\gets 1$
\li         \Else 
                $j\gets 2$               
            \End                          \label{NCG-even-j-end}
\li         \If $L[i+j-1]$ is even     \label{NCG-even-L[i+j-1]-begin}
\li         \Then
                $Q[0]\gets Q[j]^2$ 
\li         \Else
                $Q[0]\gets q Q[j]^2$    \label{NCG-even-L[i+j-1]-end}
            \End  
\li         $T[0]\gets T[j]^2-2 Q[0] I_n$  \label{NCG-even-end} 
\li     \Else                                \label{NCG-odd-begin} 
\li         \If $L[i-1]\isequal 2 L[i+1]-1$    \label{NCG-odd-j-begin}
\li         \Then 
                $j\gets 2$
\li         \Else 
                $j\gets 1$                
            \End                               \label{NCG-odd-j-end}
\li         $Q[0]\gets Q[j] Q[j+1]$          \label{NCG-odd-Q[0]}      
\li         $T[0]\gets T[j]T[j+1]-Q[0] A$   \label{NCG-odd-end}
        \End          
    \End    \label{NGC-for-end}

\li \If $L[0]$ is odd  
\label{NCG-return-begin}
\li \Then \Return ${\rm trace}(T[0])$
\li \Else      
        \Return $n(q-1)+{\rm trace}(T[0])$
 \label{NCG-return-end}
    \End  
\end{codebox}

\begin{codebox}
\Procname{$\proc{Indices}(k)$}

\li let $L$ denote a new empty array 
\li $\attrib{L}{append}(k)$  \label{INDICES-k} 
   \>\>\>\> \Comment $\attrib{L}{append}()$ means to add the parameter to the end of $L$
\zi   \>\>\>\>   \Comment the indices of $L$ start with $0$
\li \While $k$ is even   \label{INDICES-while1-begin}
\li  \Do       $k\gets \frac{k}{2}$
\li            $\attrib{L}{append}(k)$  \label{INDICES-while1-append}
    \End   \label{INDICES-while1-end}
   
\li \While $k\neq 1$    \label{INDICES-while2-begin}
\li   \Do        $k\gets \frac{k+1}{2}$
\li                 $\attrib{L}{append}(k)$  \label{INDICES-while2-append-k}
\li                 $\attrib{L}{append}(k-1)$     \label{INDICES-while2-append-k-1}
\li \If  $k$ is even
\li   \Then $k=k-1$
      \End           
    \End  \label{INDICES-while2-end}
\li \Return $L$
\end{codebox}

\medskip

The $\proc{NGC}$ procedure is a bottom-up dynamic programming algorithm. We are now going to prove the correctness of $\proc{NGC}$ and analyze its running time. 

\begin{lem}\label{lem:index}
Let $L$ denote the output array of $\proc{Indices}(k)$. 
For any entry $L[i]>1$ the following hold:
\begin{enumerate}
\item If $L[i]$ is even then $L[i+1]$ or $L[i+2]$ is equal to $\frac{L[i]}{2}$. 

\item If $L[i]$ is odd then $(L[i+1],L[i+2])$ or $(L[i+2],L[i+3])$ is equal to $\left(\frac{L[i]+1}{2},\frac{L[i]-1}{2}\right)$. 
\end{enumerate}
\end{lem}
\begin{proof}
(i): Suppose that $L[i]>1$ is even. If each of $L[0\twodots i]$ is even then $L[i+1]=\frac{L[i]}{2}$ by the first \While loop. Suppose that some of $L[0\twodots i]$ is odd. Then $L[i]$ is created by line \ref{INDICES-while2-append-k} or \ref{INDICES-while2-append-k-1}. 
If line \ref{INDICES-while2-append-k} sets $L[i]>2$ or line \ref{INDICES-while2-append-k-1} sets $L[i]$, then $L[i+2]=\frac{L[i]}{2}$ is built in the next iteration. 
If line \ref{INDICES-while2-append-k} sets $L[i]=2$ then line \ref{INDICES-while2-append-k-1} makes $L[i+1]=1$ immediately. Therefore (i) follows.

(ii): Suppose that $L[i]>1$ is odd. If line \ref{INDICES-while1-append} sets $L[i]$, then the first iteration of the second \While loop sets $(L[i+1],L[i+2])=
\left(\frac{L[i]+1}{2},\frac{L[i]-1}{2}\right)$. If line \ref{INDICES-while2-append-k}  sets $L[i]$ then $(L[i+2],L[i+3])=
\left(\frac{L[i]+1}{2},\frac{L[i]-1}{2}\right)$ is built in the next iteration. Similarly, if line \ref{INDICES-while2-append-k-1} sets $L[i]$ then $(L[i+1],L[i+2])=
\left(\frac{L[i]+1}{2},\frac{L[i]-1}{2}\right)$ is built in the next iteration. Therefore (ii) follows.
\end{proof}

In the following lemmas we state two loop invariants.

\begin{lem}\label{lem:loop-invariant-Q}
At the start of each iteration of the \For loop of lines \ref{NGC-for-begin}--\ref{NGC-for-end} of $\proc{NGC}(A,k)$, the entry  
$$
Q[0]=q^{\lfloor L[i]/2 \rfloor}.
$$
\end{lem}
\begin{proof}
By $\proc{Indices}$ procedure  the value $1$ is only stored in the last entry of $L$. 
Since line \ref{NGC-Q-initial} sets $Q[0]$ to be $1$ initially, it is true prior to the first iteration. 

To see each iteration maintains the loop invariant, let us suppose that $L[i-1]$ is even first. Lines \ref{NCG-move-begin}--\ref{NCG-move-end} move $Q[1],Q[0]$ to $Q[2], Q[1]$ respectively.  By Lemma \ref{lem:index}(i) lines \ref{NCG-even-j-begin}--\ref{NCG-even-j-end} make $L[i+j-1]=\frac{L[i-1]}{2}$. 
Thus 
$$
Q[j]=q^{\lfloor L[i+j-1]/2 \rfloor}
$$
by the loop invariant. 
Combining the above comments, lines \ref{NCG-even-L[i+j-1]-begin}--\ref{NCG-even-L[i+j-1]-end} put $q^{L[i-1]/2}$ in $Q[0]$.

Now suppose that $L[i-1]$ is odd. Lines \ref{NCG-move-begin}--\ref{NCG-move-end} move $Q[2],Q[1],Q[0]$ to $Q[3], Q[2], Q[1]$ respectively. By Lemma \ref{lem:index}(ii) lines \ref{NCG-odd-j-begin}--\ref{NCG-odd-j-end} make 
$(L[i+j-1],L[i+j])=(\frac{L[i-1]+1}{2},\frac{L[i-1]-1}{2})$. Thus  
$$
(Q[j],Q[j+1])=
\left(
q^{\lfloor L[i+j-1]/2 \rfloor},
q^{\lfloor L[i+j]/2 \rfloor}
\right)
$$
by the loop invariant. 
Combining the above comments, line \ref{NCG-odd-Q[0]} places $q^{ \lfloor L[i-1]/2 \rfloor}$ in $Q[0]$. 

Decrementing $i$ for the next iteration, the loop invariant is maintained.  The lemma follows.
\end{proof}

Applying (\ref{T(x+xinv)}) it is routine to verify that 
\begin{gather}\label{Tm+n}
T_{i+j}(x)=T_i(x) T_j(x)-T_{j-i}(x)
\qquad  
\hbox{for all integers $i,j$ with $0\leq i\leq j$}.
\end{gather}

\begin{lem}\label{lem:loop-invariant-T}
At the start of each iteration of the \For loop of lines \ref{NGC-for-begin}--\ref{NGC-for-end} of $\proc{NGC}(A,k)$, the entry
$$
T[0]=q^{\frac{L[i]}{2}} T_{L[i]}(q^{-\frac{1}{2}} A).
$$
\end{lem}
\begin{proof}
Prior to the first iteration of the loop, $T[0]=A$ and $L[i]=1$. Hence the loop invariant holds for the first time.

To see each iteration preserves the loop invariant, we suppose that $L[i-1]$ is even first. By Lemma \ref{lem:index}(i) lines \ref{NCG-even-j-begin}--\ref{NCG-even-j-end} make $L[i+j-1]=\frac{L[i-1]}{2}$. Combined with  the loop invariant, lines \ref{NCG-move-begin}--\ref{NCG-move-end} set $T[j]=q^{\frac{L[i-1]}{4}}T_{\frac{L[i-1]}{2}}(q^{-\frac{1}{2}} A)$. Hence line \ref{NCG-even-end} makes
\begin{align*}
T[0]
&=
q^{\frac{L[i-1]}{2}}T_{\frac{L[i-1]}{2}}(q^{-\frac{1}{2}} A)^2-2 Q[0] I_n
\\
&=
q^{\frac{L[i-1]}{2}}
\left(
T_{\frac{L[i-1]}{2}}(q^{-\frac{1}{2}} A)^2-2I_n  
\right)
\qquad 
\hbox{(by Lemma \ref{lem:loop-invariant-Q})}
\\
&=
q^{\frac{L[i-1]}{2}} T_{L[i-1]}(q^{-\frac{1}{2}} A)
\qquad 
\hbox{(by Equation (\ref{Tm+n}))}.
\end{align*}

Now suppose that $L[i-1]$ is odd. By Lemma \ref{lem:index}(ii), lines \ref{NCG-odd-j-begin}--\ref{NCG-odd-j-end} make $(L[i+j-1],L[i+j])=\left(\frac{L[i-1]+1}{2},\frac{L[i-1]-1}{2}\right)$. 
Combined with  the loop invariant, lines \ref{NCG-move-begin}--\ref{NCG-move-end} set 
$$
(T[j],T[j+1])=
\left(
q^{\frac{L[i-1]+1}{4}}T_{\frac{L[i-1]+1}{2}}(q^{-\frac{1}{2}} A),
q^{\frac{L[i-1]-1}{4}}T_{\frac{L[i-1]-1}{2}}(q^{-\frac{1}{2}} A)
\right).
$$ 
Hence line \ref{NCG-odd-end} makes 
\begin{align*}
T[0]
&=
q^{\frac{L[i-1]}{2}}T_{\frac{L[i-1]+1}{2}}(q^{-\frac{1}{2}} A) T_{\frac{L[i-1]-1}{2}}(q^{-\frac{1}{2}} A)
-Q[0] A
\\
&=
q^{\frac{L[i-1]}{2}}
\left(
T_{\frac{L[i-1]+1}{2}}(q^{-\frac{1}{2}} A) T_{\frac{L[i-1]-1}{2}}(q^{-\frac{1}{2}} A)-q^{-\frac{1}{2}} A
\right)
\qquad 
\hbox{(by Lemma \ref{lem:loop-invariant-Q})}
\\
&=
q^{\frac{L[i-1]}{2}} T_{L[i-1]}(q^{-\frac{1}{2}} A)
\qquad 
\hbox{(by Equation (\ref{Tm+n}))}.
\end{align*}

Decrementing $i$ for the next iteration, the loop invariant is preserved.   The lemma follows.
\end{proof}

We are now ready to prove the correctness of $\proc{NGC}$.

\begin{thm}\label{thm:correct}
$\proc{NGC}(A,k)$ returns the number $N_k$.
\end{thm}
\begin{proof}
At termination of the \For loop of lines \ref{NGC-for-begin}--\ref{NGC-for-end}, the value $i=0$. Since the input $k$ is stored in $L[0]$, it follows from  Lemma \ref{lem:loop-invariant-T} that
$$
T[0]=q^{\frac{k}{2}} T_k(q^{-\frac{1}{2}} A)
$$
after the \For loop of lines \ref{NGC-for-begin}--\ref{NGC-for-end}. 
Therefore lines \ref{NCG-return-begin}--\ref{NCG-return-end} return the number $N_k$ by (\ref{Nk&Tk}). The correctness follows.
\end{proof}

\begin{thm}\label{thm:time1}
$\proc{NGC}(A,k)$ runs in $O(n^\omega \log k)$ time.
\end{thm}
\begin{proof}
Let 
$
b_{\lfloor \log k\rfloor} b_{\lfloor \log k\rfloor-1} \cdots b_0
$
denote the binary representation of $k$. Let $h$ denote the rightmost index with $b_h=1$. Observe that 
lines \ref{INDICES-k}--\ref{INDICES-while1-end} of $\proc{Indices}(k)$ increase the length of $L$ by $h+1$ and lines \ref{INDICES-while2-begin}--\ref{INDICES-while2-end} of $\proc{Indices}(k)$ increase the length of $L$ by double of $\lfloor \log k\rfloor-h$. Hence line \ref{NGC-ell} of  $\proc{NGC}(A,k)$ sets
$$
l=2\lfloor \log k\rfloor-h+1.
$$ 
Since $n\times n$ matrix multiplication appears in $\proc{NGC}(A,k)$ exactly $l-1$ times, the time complexity  is $O(n^\omega \log k)$.
\end{proof}

Recall the number sequence $\{H_k\}_{k=1}^\infty$ from (\ref{Hk}). 
To compute $H_k$, one may modify lines~\ref{NCG-return-begin}--\ref{NCG-return-end} of $\proc{NCG}(A,k)$ to read as follows:

\begin{codebox}
\setcounter{codelinenumber}{24}
\li \If $L[0]$ is odd 
\li \Then $Q[0]\gets q^{\frac{1}{2}} Q[0]$
    \End
\li \Return $2(n-1)+Q[0]+Q[0]^{-1}-Q[0]^{-1} {\rm trace}(T[0])$
\end{codebox}

\noindent Let $\proc{Hsequence}(A,k)$ denote the resulting procedure in the rest of this paper.

\begin{thm}\label{thm:time2}
$\proc{Hsequence}(A,k)$ runs in $O(n^\omega \log k)$ time.
\end{thm}
\begin{proof}
Immediate from Theorem \ref{thm:time1}.
\end{proof}

\section{A connection between $\mu(X)$ and $\{H_k\}_{k=1}^\infty$}\label{s:proof}

\begin{lem}\label{lem:Hk_2}
For any integer $k\geq 1$ the following equation holds:
\begin{gather*}
H_k=2(n-1)-\sum_{i=2}^n T_k(q^{-\frac{1}{2}} \lambda_i).
\end{gather*}
\end{lem}
\begin{proof}
This lemms follows by substituting (\ref{Nk&Tk}) into (\ref{Hk}).
\end{proof}

\begin{lem}\label{lem1}
If $s$ is a real number with $|s|\leq 2$ then $|T_k(s)|\leq 2$ for all integers $k\geq 0$.
\end{lem}
\begin{proof}
Recall that $\frac{1}{2}T_k(2x)$ is the $k$th Chebyshev polynomial of the first kind for all integers $k\geq 0$. Hence
$$
T_k(2\cos \theta)=2 \cos k\theta
\qquad 
\hbox{for all real numbers $\theta$}.
$$
The lemma follows from the above equality.
\end{proof}

\begin{lem}
[Proposition 4.7, \cite{Huang:Li_criterion}]
\label{lem2}
Let $S$ denote a nonempty finite multiset consisting of real numbers. 
If  there is an even integer $k\geq 2$ with
\begin{gather}\label{prop1:assumption}
\frac{1}{2|S|}\sum_{s\in S} T_k(s)\leq 1
\end{gather}
then $
|s|\leq 
1+\sqrt[k]{4|S|-3}
$
for all $s\in S$.
\end{lem}

We are ready to prove the following result:

\begin{thm}\label{thm:Hsequence}
\begin{enumerate}
\item If there exists an even integer $k\geq 2$ with $H_k\geq 0$ then 
$$
\mu(X)\leq 1+\sqrt[k]{4n-7}.
$$

\item $H_k\geq 0$ for all integers $k\geq 1$ if and only if 
$$
\mu(X)\leq 2.
$$ 

\item If $\mu(X)>2$ then 
$$
\mu(X)=\lim_{k\to \infty}\sqrt{\frac{H_{2k+2}}{H_{2k}}}+\sqrt{\frac{H_{2k}}{H_{2k+2}}}.
$$
\end{enumerate}
\end{thm}
\begin{proof}
(i): Suppose that there exists an even integer $k\geq 2$ with $H_k\geq 0$. By Lemma \ref{lem:Hk_2}, it is equivalent to 
$$
\frac{1}{2(n-1)}\sum_{i=2}^n T_k(q^{-\frac{1}{2}}\lambda_i)\leq 1.
$$
Applying Lemma \ref{lem2} it follows that $q^{-\frac{1}{2}}|\lambda_i|\leq 1+\sqrt[k]{4n-7}$ for all $i=2,3,\ldots,n$. Therefore (i) follows.

(ii): ($\Rightarrow$) Since $H_k\geq 0$ for all even integers $k\geq 2$, it follows from Theorem \ref{thm:Hsequence}(i) that 
$$
\mu(X)\leq \lim_{k\to\infty} 1+\sqrt[k]{4n-7}=2.
$$ 

($\Leftarrow$) Since $\mu(X)\leq 2$ it follows that $q^{-\frac{1}{2}}|\lambda_i|\leq 2$ for all $i=2,3,\ldots,n$. Combined with Lemma \ref{lem1} this yields that 
$|T_k(q^{-\frac{1}{2}}\lambda_i)|\leq 2$
for all $i=2,3,\ldots,n$ and all integers $k\geq 0$.
Hence $H_k\geq 0$ for all $k\geq 1$ by Lemma \ref{lem:Hk_2}.

(iii): 
Recall that $\alpha_i$ is a root of $u^2-q^{-\frac{1}{2}} \lambda_i u+1$ for all $i=1,2,\ldots,n$. Observe that if $q^{-\frac{1}{2}} |\lambda_i|\leq 2$ then $|\alpha_i|=1$.
Since $\mu(X)>2$ there is an $i\in\{2,3,\ldots,n\}$ with $q^{-\frac{1}{2}}|\lambda_i|>2$ and hence $\alpha_i$ is a real number with $\alpha_i\not= \pm 1$. Therefore 
$$
\beta=\max_{2\leq i\leq n}\{|\alpha_i|,|\alpha_i|^{-1}\}>1.
$$
Since $f(x)=x+x^{-1}$ is strictly increasing on $(1,\infty)$, we have 
$$
\beta+\beta^{-1}=\max_{2\leq i\leq n} |\alpha_i|+|\alpha_i|^{-1}=\mu (X).
$$
By (\ref{T(x+xinv)}) we have 
$$
\lim_{k\to \infty}
\frac{T_{2k}(\alpha_i+\alpha_i^{-1})}{\beta^{2k}}
=
\lim_{k\to \infty}
\frac{\alpha_i^{2k}+\alpha_i^{-2k}}{\beta^{2k}}
=\left\{
\begin{array}{ll}
1
\qquad 
&\hbox{if $\beta=\max\{|\alpha_i|,|\alpha_i|^{-1}\}$},
\\
0
\qquad 
&\hbox{else}
\end{array}
\right.
$$
for all $i=2,3,\ldots,n$. Combined with Lemma \ref{lem:Hk_2} this yields that $H_{2k}$ is asymptotic to $-m \beta^{2k}$ as $k$ approaches to $\infty$ where $m$ is the number of $i\in\{2,3,\ldots,n\}$ with $\beta=\max\{|\alpha_i|,|\alpha_i|^{-1}\}$. By the above comments the statement (iii) follows.
\end{proof}

\section{A fast algorithm to evaluate $\mu(X)$}\label{s:algorithm}

With the inputs $A$ and a real number $\varepsilon>0$, Theorem \ref{thm:Hsequence} yields the following  algorithm that return $\const{true}$ if $\mu(X)$ or its estimation is less than or equal to $2+\varepsilon$ and return $\const{false}$ else:

\begin{codebox}
\Procname{$\proc{SpectralExpansion}(A,\varepsilon)$}

\li $n\gets \hbox{the number of rows of $A$}$

\li $k\gets 2 \lceil \log (4n-7)/2 \log (1+\varepsilon)\rceil$ \label{SE:k}

\li $k'\gets k+2$ \label{SE:k'}
    
\li $H\gets \proc{Hsequence}(A,k)$ \label{SE:H}

\li $H'\gets \proc{Hsequence}(A,k')$ \label{SE:H'}

\li \If $H \geq 0$ or $H'\geq 0$
\li \Then 
          \Return \const{true}
\li \Else  
       $\lambda\gets  \sqrt{H'/H}+\sqrt{H/H'}$  
\li    print $\lambda$ `` is an estimation of $\mu(X)$.''           
\li    \If $\lambda\leq 2+\varepsilon$
\li     \Then  \Return \const{true}
\li    \Else  
              \Return \const{false}
        \End
    \End

\end{codebox}

\begin{thm}
$\proc{SpectralExpansion}(A,\e)$ runs in $O(n^\omega \log\log_{1+\e} n)$ time.
\end{thm}
\begin{proof}
By line \ref{SE:k} the value 
$k=2\left \lceil \frac{\log_{1+\e} (4n-7)}{2} \right \rceil=O(\log_{1+\e} n)$.
Hence the lines \ref{SE:H} and \ref{SE:H'} take $O(n^\omega \log\log_{1+\e} n)$ time by Theorem \ref{thm:time2}. The other lines run in $O(1)$ time. The result follows.
\end{proof}

We have a remark on the procedure $\proc{SpectralExpansion}$. 
In the case of $H<0$ and $H'< 0$, if the parameter $\e$ is small enough $\proc{SpectralExpansion}(A,\e)$ gives a nice estimation of $\mu(X)$ theoretically; however, if $\e$ is not small enough the procedure $\proc{SpectralExpansion}(A,\e)$ perhaps returns incorrect information.

\subsection*{Acknowledgements}
The research is supported by the Ministry of Science and Technology of Taiwan under the project MOST 106-2628-M-008-001-MY4.

\bibliographystyle{amsplain}
\bibliography{MP}

\end{document}